\documentclass{article}

\usepackage{colortbl}
\usepackage[table]{xcolor}%

\usepackage{microtype}%

\usepackage{array}
\usepackage{booktabs}
\usepackage{multirow}
\usepackage{diagbox}
\usepackage{pict2e} %
\usepackage{makecell}
\usepackage{xspace}
\usepackage{paralist}
\usepackage{tikz}

\newcommand{\Alice}{Akira\xspace}
\newcommand{\Charlie}{Cobra\xspace}
\newcommand{\Ken}{Ken\xspace}
\newcommand{\Kakarotto}{Kakarotto\xspace}
\newcommand{\Tom}{Totoro\xspace}
\newcommand{\N}{\ensuremath{\mathbb N}}

\newcommand{\seck}{\mathfrak{K}}

\title{Physical Zero-Knowledge Proofs for Akari, Takuzu, Kakuro and
  KenKen \footnote{This work was partially supported by “Digital
    trust” Chair from the University of Auvergne Foundation, by the HPAC project of the French Agence Nationale de la Recherche (ANR~11~BS02~013), and by the CNRS PEPS JCJC VESPA. The authors would like to thank the anonymous reviewers for their helpful comments.}}

\usepackage{color,svgcolor}
\usepackage[utf8]{inputenc}
\usepackage{lmodern,amsmath}
\usepackage[T1]{fontenc}
\usepackage{textcomp}
\usepackage[mathscr]{eucal}
\usepackage{amssymb,amsthm}
\usepackage[plainpages=true]{hyperref}
\usepackage[perpage,symbol*]{footmisc}
\DefineFNsymbols*{Pascal}{{\ensuremath{\star}}{\ensuremath{\diamond}}{\ensuremath{\dagger}}{\ensuremath{\ddagger}}{{\ensuremath{\sharp}}}{{\ensuremath{\triangle}}}{{\ensuremath{\dagger\dagger}}}{{\ensuremath{\ddagger\ddagger}}}}
\setfnsymbol{Pascal}

\makeatletter
\hypersetup{
pdftitle={Physical Zero-Knowledge Proofs for Akari, Takuzu, Kakuro and KenKen},
pdfauthor={Bultel, Dreier, Dumas, Lafourcade},
breaklinks=true,
colorlinks=true,
 linkcolor=darkred,
 citecolor=blue,
 urlcolor=darkgreen,
}
\makeatother

\author{Xavier Bultel\footnotemark[2]
\and Jannik Dreier \footnotemark[3]\,\,\footnotemark[4]\ \footnotemark[5]
\and Jean-Guillaume Dumas \footnotemark[6]
\and Pascal Lafourcade\footnotemark[2]}

\newtheorem{lemma}{Lemma}

\begin{document}

\maketitle
\footnotetext[2]{LIMOS, University Clermont Auvergne,
  Campus des C\'ezeaux, Aubi\`ere, France.\\
  \href{mailto:xavier.bultel\string
    @udamail.fr,pascal.lafourcade\string @udamail.fr}{firstname.lastname\string @udamail.fr}}
\footnotetext[3]{Universit\'e de Lorraine, Loria, UMR 7503, Vandoeuvre-l\`es-Nancy, France.}
\footnotetext[4]{Inria, Villers-l\`es-Nancy, France.}
\footnotetext[5]{CNRS, Loria, UMR 7503, Vandoeuvre-l\`es-Nancy, France.\\
  \href{mailto:jannik.dreier\string
    @loria.fr}{jannik.dreier\string @loria.fr}}
\footnotetext[6]{LJK, Universit\'e Grenoble Alpes, CNRS umr 5224,
  51, av. des Mathématiques, BP53, 38041 Grenoble, France.\\
  \href{mailto:Jean-Guillaume.Dumas\string
    @imag.fr}{Jean-Guillaume.Dumas\string @imag.fr}}

\begin{abstract}
Akari, Takuzu, Kakuro and KenKen are logic games similar to Sudoku. 
In Akari, a labyrinth on a grid has to be lit by placing lanterns,
respecting various constraints.  In Takuzu a grid has to be filled
with 0's and 1's, while respecting certain constraints.  In Kakuro a
grid has to be filled with numbers such that the sums per row and
column match given values; similarly in KenKen a grid has to be filled
with numbers such that in given areas the product, sum, difference or
quotient equals a given value.  We give physical algorithms to realize
zero-knowledge proofs for these games which allow a player to show
that he knows a solution without revealing it.  These interactive
proofs can be realized with simple office material as they only rely
on cards and envelopes.  Moreover, we formalize our algorithms and
prove their security.

 \end{abstract}

\section{Introduction}

\Alice and \Tom, close friends and both great fans of logical games, decide to challenge each other on their favorite games.
\Alice is an expert in Akari, and \Tom is a specialist of Takuzu.
Each one proposes a grid of his favorite game to the other one as a
challenge.  However, as both are extremely competitive, they choose
grids that are so difficult that the other is not able to solve them,
as they are less experienced in the other game.
\Tom, working in security, immediately supposes that something went wrong, and wants a proof from \Alice that his grid actually has a solution.
\Alice, feeling hurt by this distrust, directly asks \Tom the same question.
Obviously both of them do not want to reveal the solution, as this
would render the challenge pointless. So they decide to ask \Ken, a
common friend, for help.  However, \Ken, a grand master of KenKen, has
the same problem.  He and his best friend \Kakarotto exchanged
challenges for KenKen and Kakuro, and want to show each other that
they know the solution without revealing it.
But then \Kakarotto has the idea to ask \Charlie, a computer
scientist and cryptographer and common friend, for help.
\Charlie directly sees a solution: a zero knowledge proof of knowledge of the solution (ZKP).
A ZKP is a protocol that allows a prover $P$ to convince a verifier $V$ that he knows some solution $s$ to the instance $\mathcal{I}$ of a problem $\mathcal{P}$, without revealing any information about $s$.
Such a protocol satisfies three properties\footnote{Moreover, if $\mathcal{P}$ is NP-complete, then the ZKP should be polynomial.
Otherwise it might be easier to find a solution than proving that a solution is a correct solution, making the proof pointless.}:
\begin{compactdesc}
\item[Completeness:] If $P$ knows $s$, then he is able to convince $V$.
\item[Soundness:] If $P$ does not know $s$, then he is not able to convince $V$ except with some \emph{``small''} probability\footnote{More precisely, we want a negligible probability, \emph{i.e.}, the probability should be a function $f$ of a security parameter $\seck$ (for example the number of repetitions of the protocol) such that $f$ is negligible, that is for every polynomial $P$, there exists  $n_0 > 0$ such that  $\forall~ x > n_0,f(x)<1/P(x)$.}.
\item[Zero-knowledge:] $V$ learns \emph{nothing} about $s$ except $\mathcal{I}$, \emph{i.e.} there exists a probabilistic polynomial time algorithm $\texttt{Sim}(\mathcal{I})$ (called the simulator) such that outputs of the real protocol and outputs of $\texttt{Sim}(\mathcal{I})$ follow the same probability distribution.
\end{compactdesc}
Yet, ZKPs are usually executed by computers, which \Tom (a very skeptical security expert) does not trust.
However, \Charlie is still able to help them, by inventing a ZPK that
relies only on physical objects such as cards and envelopes.
In this paper, we present \Charlie's solution to \Alice's, \Tom's, \Ken's and \Kakarotto's dilemma.%

\subparagraph{Contributions:}
We construct physical ZKP for Akari, Kakuro, KenKen and Takuzu.
\begin{compactitem}
\item For Akari, our ZKP construction uses special cards and envelopes.
 Moreover, the prover needs to interact with the verifier to construct
 the proof.
\item For Takuzu, we propose an interactive construction using cards, paper and envelopes.
\item For Kakuro, we use red and black cards and envelopes to implement an addition operation.
\item For KenKen, we also rely on red and black cards and envelopes, but the interactive proof is more complex due to the different operations (sum, difference, product, quotient).
\end{compactitem}
We also prove the security of our constructions.%
\subparagraph{Related Work:}
Sudoku, introduced under this name in 1986 by the Japanese puzzle
company Nikoli, and similar games such as Akari, Takuzu, Kakuro and
Ken-Ken have gained immense popularity in recent years.  Following the
success of Sudoku, generalizations such as Mojidoku which uses letters
instead of digits, and other similar logic puzzles like Hitori, Masyu,
Futoshiki, Hashiwokakero, or Nurikabe were developed.  Many of them
have been proved to be
NP-complete~\cite{journals/icga/KendallPS08,DBLP:conf/mfcs/Demaine01}.

Interactive ZKPs were introduced by Goldwasser \emph{et al.}~\cite{Goldwasser:1985:KCI:22145.22178}, and it was then shown that for any NP-complete problem there exists an interactive ZKP protocol~\cite{GoldreichNP}.
An extension by Ben-Or \emph{et al.}~\cite{Ben-Or:1990:EPP:88314.88333} showed that every provable statement can be proved in zero-knowledge.
They also showed that physical protocols using envelopes exist, yet their construction is -- due to its generality -- rather involved an often impractical for real problem instances.
Proofs can also be non-interactive in the sense that the prover and verifier do not need to interact during the protocol~\cite{Blum:1988:NZA:62212.62222}.
For more background on ZKPs see for example~\cite{Menezes:1996:HAC:548089}.

As ZKPs have always been difficult to explain, there are works on how to explain the concepts to non experts, partly using physical protocols as illustrations. 
For example, in their famous paper~\cite{Quisquater:1989:EZP:646754.705056}, Quisquater \emph{et al.} propose ``Ali Bababa's cave'' as a tool to explain Zero-Knowledge Proof to kids.
In~\cite{Shasha:2014:UPP:2684442.2670917}, ZKP's are illustrated using a magician that can count the number of sand grains in a heap of sand.
Naor \emph{et al.} used the well-known ``Where's Waldo?'' cartoons to explain the concept of ZKP to kids, and also proposed an efficient physical protocol for the problem in~\cite{Naor99appliedkid}.

In 2007, the same authors proposed a ZKP for Sudoku using cards~\cite{Gradwohl:2007:CPZ:1760607.1760623}, which partly inspired \Charlie's solution in our paper. This was later extended for Hanjie~\cite{ChienH10}.

As in~\cite{Gradwohl:2007:CPZ:1760607.1760623}, we here also assume an abstract {\em shuffle functionality} 
which is essentially an indistinguishable shuffle of a set of sealed envelopes
or of face down cards.
This functionality is necessary to prevent information leakage, and cannot be realized neither by the verifier nor the prover.
The verifier cannot perform this action, as otherwise he could be perform a shuffle that is not random, and break zero-knowledge.
Moreover, in a physical protocol the prover cannot perform this action either, as he might be able to tamper with the packets at this step, similar to a magician performing a card trick.
A possible realization would be to rely on a trusted third party (\Charlie for instance), or trying to ensure that the prover cannot mess with the cards.

Moreover, there is work on implementing cryptographic protocols using
physical objects, for example in~\cite{DBLP:conf/fun/MizukiS14}, where
the authors use cards to perform multi-party computations, or
in~\cite{DBLP:conf/fun/DreierJL14} where a physical secure auction
protocol is proposed.%

\subparagraph*{Outline:}
For all games, we first present the rules followed by our ZKP construction and then the security analysis. In Section~\ref{sec:akari}, we present Akari, in Section~\ref{sec:takuzu}, 
Takuzu, in Section~\ref{sec:kakuro} Kakuro, and in
Section~\ref{sec:kenken} KenKen. In
Section~\ref{sec:conclusion}, we conclude the paper.

\section{Akari}\label{sec:akari}

\emph{Akari} is a logic puzzle first published in 2001 by
\emph{Nikoli}, a Japanese games and logic puzzles publisher. It is
also called \emph{Light up}\footnote{\url{https://en.wikipedia.org/wiki/Light_Up_(puzzle)}}. The
goal of the game is to illuminate all the cells of a rectangular grid
by placing some lights. The lights illuminate horizontally and vertically all
adjacent white cells. The grid also contains black cells, which represent walls that the light cannot cross.
Some  black cells can contain an integer between 0 and 4.
To illuminate all the cells, lights have to be added in the white
cells according to the following constraints:
\begin{compactenum}
\item Two lights cannot illuminate each other.
\item The number in a black cell indicates how many lights are present in the adjacent white cells, \emph{i.e.}, in the white cells directly above, below, left and right of the number.
\item All the white cells are illuminated by at least one light.
\end{compactenum}
In Figure~\ref{fig:akari}, we give a simple example of an Akari game.
It is easy to verify that the three constraints are satisfied, but solving Akari was shown to be NP-complete \emph{via} a reduction from Circuit-SAT~\cite{akariNP05}.

\begin{figure}[bt]
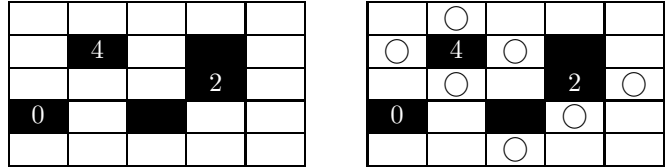

\begin{center}
  \begin{tabular}{|c|c|c|c|c|}
\hline
 {\color{white}$\bigcirc$} & {\color{white}$\bigcirc$} & {\color{white}$\bigcirc$} & {\color{white}$\bigcirc$} & {\color{white}$\bigcirc$} \\ \hline
   &\textcolor{white}{\cellcolor{black}{4}} & &\cellcolor{black} & \\ \hline
   & & & \textcolor{white}{\cellcolor{black}{2}} & \\ \hline
 \textcolor{white}{\cellcolor{black}{0}} & &\cellcolor{black} & & \\ \hline
   & & & & \\ \hline
\end{tabular}
\qquad
\begin{tabular}{|c|c|c|c|c|}
  \hline
    & $\bigcirc$ &   &  &   \\ \hline
 $\bigcirc$ &\textcolor{white}{\cellcolor{black}{4}}&  $\bigcirc$ &\cellcolor{black}  &   \\ \hline
    & $\bigcirc$ &   & \textcolor{white}{\cellcolor{black}{2}}  &  $\bigcirc$ \\ \hline
  \textcolor{white}{\cellcolor{black}{0}}&   &\cellcolor{black}  & $\bigcirc$ & \\ \hline  &  & $\bigcirc$   &  &   \\ \hline
\end{tabular}
\end{center}
\caption{Example of an Akari grid and its solution.}\label{fig:akari}
\end{figure}

\subsection{ZKP Construction for Akari}

Our aim is to give a ZKP proof such that \Alice, the prover denoted by
$P$, proves to \Tom, the verifier denoted by $V$, that he knows a
solution of a given Akari grid $G$. We assume that the grid is printed on
a sheet of paper, and that we can use a \emph{shuffle 
functionality}. We use two kinds of cards to model the fact that
there is one light or not in one cell of the solution. The symbol
$\bigcirc$ represents the presence of a light, and an empty card the
absence of a light.%
\subparagraph*{Setup:}
According to his solution, the prover $P$ places, face
down on each white cell, $\bigcirc$ cards when there is a light, and empty cards
when there is none.  Note that all the cards in the packet on the same cell are
thus identical.  Placing the cards happens in three phases:
\begin{compactenum}
\item In all white cells, $P$ adds 2 cards face down (first illustration in
  Figure~\ref{fig:akariproof}).
\item $P$ adds an additional card for all white cells adjacent to a black cell that contains a number.
  For instance, in the second illustration in Figure~\ref{fig:akariproof}, for the black cell containing the number $4$, $P$ adds one card in all positions with a grey background, and similarly for the black cells containing $0$ and $2$.
\item Finally, for each white cell in the grid, $P$ places a card on the cell and all horizontally and vertically adjacent white cells.
  The third illustration in Figure~\ref{fig:akariproof} shows this for the top left cell: $P$ adds a card to the cell itself, and to the six adjacent cells (all marked in grey, two below and four to the right).
  This has to be done for all white cells.
\item The same three steps are repeated once more to build a dual second grid, \emph{i.e.}, a grid
  where all empty cards are replaced by light cards and \emph{vice-versa}.
\end{compactenum}
At the end of the algorithm, $P$ has added multiple cards face
down on each cell of the grid $G$ and its dual grid $\hat{G}$.%

\begin{figure}[tb]
  \begin{center}
\begin{tabular}{|c|c|c|c|c|}
 \hline
 $2$ & $2$ & $2$ & $2$ & $2$ \\ \hline
 $2$ &\textcolor{white}{\cellcolor{black}{4}}& $2$ &\cellcolor{black} & $2$ \\ \hline
 $2$ & $2$ & $2$ & \textcolor{white}{\cellcolor{black}{2}} &  $2$ \\ \hline
  \textcolor{white}{\cellcolor{black}{0}}& $2$ & \cellcolor{black} & $2$ & $2$ \\ \hline
 $2$ & $2$ & $2$ & $2$ & $2$ \\ \hline
\end{tabular}
\begin{tabular}{|c|c|c|c|c|}
  \hline
  & \cellcolor{gray!30}$1$ &  &   & {\color{white}$1$} \\ \hline
 \cellcolor{gray!30}$1$ &\textcolor{white}{\cellcolor{black}{4}}& \cellcolor{gray!30}$1$ &\cellcolor{black} &  \\ \hline
  & \cellcolor{gray!30}$1$ &  & \textcolor{white}{\cellcolor{black}{2}} &  \\ \hline
  \textcolor{white}{\cellcolor{black}{0}}&  & \cellcolor{black} &  &  \\ \hline
  &  &  &  &  \\ \hline
\end{tabular}
\begin{tabular}{|c|c|c|c|c|}
  \hline
 \cellcolor{gray!30}$1$ & \cellcolor{gray!30}$1$ & \cellcolor{gray!30}$1$ &  \cellcolor{gray!30}$1$ & \cellcolor{gray!30}$1$ \\ \hline
 \cellcolor{gray!30}$1$ &\textcolor{white}{\cellcolor{black}{4}}&  & \cellcolor{black} &  \\ \hline
 \cellcolor{gray!30}$1$ &  &  & \textcolor{white}{\cellcolor{black}{2}} &  \\ \hline
  \textcolor{white}{\cellcolor{black}{0}} &  & \cellcolor{black} &  & \\ \hline
  &  &  &  &  \\ \hline
\end{tabular}
  \end{center}
  \caption{Illustrations for steps 1 to 3 in the Akari ZKP
    construction. From left to right: {\bf Step 1:} two cards added on all
    white cells. {\bf Step 2:} cards added around black cells with
    number. {\bf Step 3:} cards added per cell as a function of
    adjacent cells.}\label{fig:akariproof}
\end{figure}

\subparagraph*{Verification:}
To check the correctness of the solution, the verifier picks a random bit $c$ in $\{0,1\}$ and executes the following verifications:
\begin{compactdesc}
\item[if $c=0$:] For each cell of the grid $G$ the verifier places the packet of
  cards in an envelope, then he does the same for each cell of grid
  $\hat{G}$. He shuffles the envelopes using the shuffle functionality, and then
  opens each envelope to check whether it contains only light cards or only
  empty cards.  With this he is sure that all the cards on
  a given cell were identical.
\item[if $c=1$:] The verifier discards all the cards on the dual grid
  $\hat{G}$ and checks whether the $3$ constraints that a solution should 
  satisfy are respected using the cards on the grid $G$:
\begin{compactenum}
\item {\bf No two lights see each other:} For each row or
  column of adjacent white cells, $V$ randomly picks one
  card per cell to form a packet.
  Then $V$ passes the cards to the \emph{shuffle functionality} who 
  shuffles them, and passes them on to $P$. If all cards
  are empty cards then $P$ adds a $\bigcirc$ card, otherwise he adds
  an empty card. $P$ returns the cards, once again shuffled, to $V$, who checks that there
  is exactly one light card in the packet. 
\item {\bf Black cells:} For each black cell that contains a number
  $l$, $V$ picks one card in all adjacent cells and gives them to the
  shuffle functionality. He then looks at the cards and checks that
  the number $l$ corresponds to the number of $\bigcirc$ cards. In
  second illustration of Figure~\ref{fig:akariproof}, a card from all
  grey cells has to be collected to verify the black cell containing
  the number $4$.  %
\item {\bf All cells are illuminated:} For each cell, $V$ picks one
  card from the cell plus one card from all horizontally and
  vertically adjacent white cells (for example, for the top left cell
  in Figure~\ref{fig:akariproof}, he takes one card from all grey
  cells).  
  $V$ passes all collected cards face down to~$P$.  If there are two light cards
  in the received cards, then $P$ adds one empty card, otherwise he adds one
  $\bigcirc$ card. $P$ passes the cards to the shuffle functionality, and
  returns them to $V$.
  Then $V$ opens the cards and checks that there is exactly two light cards in
  the packet.  %
\end{compactenum}
\end{compactdesc}
This protocol is repeated $\seck$ times where $\seck$ is a chosen security
parameter.  Note that the number of cards to place on the grid and the number of
verification steps is polynomial in the size of the grid, making our ZKP
polynomial both for the prover and the verifier.

\subsection{Security Proofs for Akari}
We now prove the security of our construction.

\begin{lemma}[Akari Completeness]
If $P$ knows a solution of a given Akari grid, then he is able to convince $V$.
\end{lemma}
  \begin{proof}
If $P$ knows the solution of an Akari grid, he can place the cards on the grids
($G$ and $\hat{G}$) following our ZKP algorithm. Then on each cell of both grids
there is a packet of identical cards, showing either nothing or a light,
depending on whether there is a light in the solution. The verifier picks $c$
among $\{0,1\}$.
\begin{compactdesc}
\item[if $c=0$:] The verifier $V$ puts each packet on each cell for the two grids 
  into envelopes and shuffles them. Since all the packets on cells on both
  grids contain the same kind of card, the verifier accepts.
\item[if $c=1$:] In the following, when taking a card from a cell, the verifier $V$ 
  randomly picks one card among all the cards in that cell. This has no effect when 
  the prover is honest since all of them are identical. Using
  these cards, the verifier checks the three properties that a solution should satisfy.
\begin{compactitem}
\item To check that two lights cannot see each other, for all lines and columns
  the verifier takes one card from each cell inside the line or column
  {and asks the prover to add one extra card}.
    The resulting selection is then shuffled. 
  As the solution is correct, 
  there can be either no light or one light, but not more, as otherwise two 
  lights would see each other. If there is none, $P$ adds one, 
  otherwise he adds an empty card. Thus the cards contain exactly one light, 
  which $V$ will successfully verify.
\item For all black cells that contain a number, the verifier picks
  one card in all adjacent white cells, 
uses the shuffle functionality on the selected cards,
  and 
  then verifies that they contain the same number of lights as is
  written in the black cell. As $P$'s solution is correct, the 
  verification will succeed.
\item To verify that each cell of the grid is illuminated by a
  light, for each cell the verifier picks one card from all white cells
  which are horizontally and vertically adjacent. 
  He asks $P$ to add exactly one card according to the algorithm and uses the
  shuffle functionality on the remaining cards.  As the solution is correct,
  there will be either one or two lamps in the packet since each cell is
  illuminated, so after adding the right card there will be exactly two. Then
  $V$ verifies that the set of cards contains exactly two lights, which will
  succeed.
\end{compactitem}
\end{compactdesc}
\vspace{-\baselineskip}
\end{proof}

\begin{lemma}[Akari Soundness]
If $P$ does not know a solution of a given Akari's grid of size $n$ then he is
not able to convince $V$ except with a negligible probability lower than
$p=(\frac{1}{2})^\seck$ when the protocol is repeated $\seck$ times.
\end{lemma}

\begin{proof}
We call a \emph{cell-packet} the set of cards that are placed in one cell. When
the verifier is collecting some cards during the verification phase for one
cell, we also call a \emph{row-packet} the set of cards that are collected in an
horizontal line, a \emph{column-packet} the set of cards that are collected in a
vertical line and \emph{cross-packet} the set of cards that are collected in
horizontally and vertically adjacent white cells of a given cell.

Given a grid of size $n$, we show that if $P$ is able to perform the proof for
  any challenge $c$, then he has a solution of the grid. This implies by
  contraposition that if the prover does not have a solution, then he fails the
  procedure for either $c = 0$, or $c = 1$, or both. In that case the
  probability that the verifier asks him a challenge for which he does not have
  a solution is at least $\frac{1}{2}$, which results in the probability of
  $p=(\frac{1}{2})^\seck$ as the protocol is repeated $\seck$ times.

Thus it suffices to show that if $P$ is able to perform the proof for $c=0$ and $c=1$, then he has a solution of the grid.

If $P$ is able to perform the proof for $c=0$, then his commitment is
well-formed, \emph{i.e.}, all the cards in each cell-packet are the same.

If $P$ is able to perform the proof for $c=1$, then his solution passes all the
verifier's checks. Moreover, we show that if the commitment of the prover is
well-formed (which we know from the above), and this commitment does not
correspond to a solution, then the verifier detects it with our ZKP
algorithm. This implies that the prover has a solution to the grid, which is
what we need to show to conclude the proof.

Suppose now that the commitment is well-formed, and that it does not correspond to a 
solution. Then we can distinguish three cases for the commitment, each corresponding to one constraint of the game that is not respected:
\begin{compactitem}
\item There are two lights that can see each other. According to our ZKP, the
  verifier picks a row-packet or a column-packet for each row or column and passes it to
  the prover. The prover adds one card and gives it back to the verifier. Then
  the verifier checks that there is exactly one light in the packet. If there are two
  lights that can see each other, the prover cannot cheat and the verifier will see 
  at least two lights.
\item There is not the right number of lights close to a black cell containing a
  number. Then the verifier discovers that the number of lights does not match when he
  performs the second point of our ZKP algorithm.
\item A cell is not illuminated. According to our algorithm, the verifier picks
  a cross-packet for this cell. Then he gives the packet to the prover who adds
  one extra card. Then the verifier checks that there are exactly two lights in
  the packet. If the cell is not illuminated, he will have at most one light in
  the packet and thus detect the error.
\end{compactitem}
Thus, if the commitment is well-formed and does not correspond to a
solution, the verifier will detect it. Hence, if the prover passes all the
verifier's checks, he has a solution, which concludes the proof.
\end{proof}
Note that an optimal soundness of $0$ can be obtained using a stronger
model, namely the \emph{triplicate functionality}
of~\cite[\S~4.2]{Gradwohl:2007:CPZ:1760607.1760623}. Instead of the challenge
$c=0$, one uses cards that can be cut in $k$ equal parts (a ``\emph{$k$-plicate functionality}''), thus guaranteeing identical parts.
\begin{lemma}[Akari Zero-Knowledge]\label{lem:akari:zk}
$V$ learns nothing about $P$'s solution of a given Akari grid.
\end{lemma}
\begin{proof}
We use the proof technique described in~\cite[Protocol~3]{Gradwohl:2007:CPZ:1760607.1760623}: to show zero-knowledge,
we have to describe an efficient simulator that simulates any interaction between
a cheating verifier and a real prover. The simulator does not have a correct
solution, but it does have an advantage over the prover: when shuffling decks, it
is allowed to swap the packets for different ones. We thus show how to
construct a simulator for each challenge $c$ in $\{0,1\}$ in the Akari ZKP.
\begin{compactdesc}
\item[if $c=0$:] The simulator simulates the prover $P$ as
  follows: it fills the grid $G$ only with empty cards and the dual 
  grid $\hat{G}$ only with light cards. If the verifier chooses the challenge
  $c=0$, the verifier puts each packet of cards in an envelope of both grids and 
  shuffles them. Then he checks that all packets contain only the same card.
  This perfectly simulates an interaction with the real prover because there are 
  two packets for each cell: one with only light cards and one with only empty cards.

  \item[if $c=1$:] The simulator first prepares the following packets:
\begin{compactitem}
\item For each white cell in a row of $r$ white cells and in a column of $c$  white cells, the simulator sets:
  \begin{compactitem}
  \item A packet containing  $(c)$ empty cards (denoted type $1$).
  \item A packet containing  $(r)$ empty cards (denoted type $2$).
  \item A packet containing  $(c+r-3)$ empty cards and $2$ $\bigcirc$ cards (denoted type $3$).
  \end{compactitem}
 \item For each black cell with number $l$, the simulator sets a packet containing $l$ $\bigcirc$ cards and $4-l$ empty cards. Such a packet is of type $4$.
\end{compactitem}

Then, still when the verifier has chosen the challenge $c=1$, the verifier
discards the dual grid $\hat{G}$ and execute the following checks:
\begin{compactitem}
\item {\bf No two lights see each other:} for each horizontal (resp.
  vertical) line of adjacent white cells, $V$ randomly picks one
  card per cell to form a packet.
  Then the simulator, acting as the the shuffle functionality, replaces this
  packet by the type $1$ (resp. $2$) packet corresponding to the same cell. 
  The simulator, now acting as $P$, adds one $\bigcirc$ card to the packet and returns the cards to $V$, who can check that there is exactly one light card in the packet. 
  
\item {\bf Black cells:} for each black cell that contains a number $l$, $V$
  picks one card in all adjacent cells to form a packet. Then the simulator,
  acting as the shuffle functionality, replaces this packet by the type
  $4$ packets corresponding to the same cell. $V$ looks at the cards and can check
  that $l$ corresponds to the number of $\bigcirc$ cards.
\item {\bf All cells are illuminated:} for each cell, $V$ picks one
  card from the cell plus one card from all horizontally and
  vertically adjacent white cells. 
 Then the simulator, acting as the shuffle functionality, replaces this packet
 by the type $3$ packet corresponding to the same cell.
 The simulator, now acting as $P$, adds one empty card to the packet
 and returns the cards to $V$, who can check that there are exactly two light cards
 in the packet.
\end{compactitem}
For each cell, all cards are face down before the last shuffle of all
packets and the simulated proofs and the real proofs are indistinguishable.
Therefore, $V$ learns nothing from the verification phases and the protocol is zero-knowledge.
\end{compactdesc}
\vspace{-\baselineskip}
\end{proof}

\section{Takuzu}\label{sec:takuzu}

\emph{Takuzu} is a puzzle invented by Frank Coussement and Peter De
Schepper in 2009\footnote{\url{https://en.wikipedia.org/wiki/Takuzu}}. 
It was also called \emph{Binero}, \emph{Bineiro}, \emph{Binary
  Puzzle}, \emph{Brain Snacks} or \emph{Zernero}. A similar game
(using crosses and circles) was proposed in 2012 by Aldolfo Zanellat
under the name of \emph{Tic-Tac-Logic}.  The goal of Takuzu is to fill a rectangular
grid of even size with 0's and 1's. An initial Takuzu grid already contains a 
few filled cases. A grid is solved when it is full (\emph{i.e.}, no empty cases)
and respects the following constraints:
\begin{compactenum}
\item Each row and each column contains exactly the same number of 1's
  and 0's.
\item Each row is unique among all rows, and each column is unique
  among all columns.
\item In each row and each column there can be no more than two same
  numbers adjacent to each other; for example $110010$ is possible,
  but $110001$ is impossible.
\end{compactenum}

Figure~\ref{fig:takuzu} contains a simple $4 \times 4$ Takuzu grid and
its solution. We can verify that each row and column is unique,
contains the same number of 0's and 1's, and there are never three
consecutive 1's or 0's. The problem of solving a Takuzu grid was 
proven to be NP complete in~\cite{BPNP12,BPNP15}.

\begin{figure}[tb]
\begin{center}
\begin{tabular}{|c|c|c|c|}
  \hline
   & {\bf 1} &  & {\bf 0}   \\ \hline
   &  & {\bf 0}&    \\ \hline
   & {\bf 0} &  &    \\ \hline
  {\bf 1} & {\bf 1} &  & {\bf 0}   \\ \hline
\end{tabular}\qquad
\begin{tabular}{|c|c|c|c|}
  \hline
   0 & {\bf 1} & 1  & {\bf 0}   \\ \hline
   1 &  0 & {\bf 0}& 1   \\ \hline
   0 & \bf{0} & 1 & 1   \\ \hline
  {\bf 1} & {\bf 1} & 0 & {\bf 0}   \\ \hline
\end{tabular}
\end{center}
\caption{Example of a Takuzu grid and its solution.}\label{fig:takuzu}
\end{figure}

\subsection{ZKP Construction for Takuzu}\label{ssec:takuzuzkp}

Similar to Akari, our aim is to give a ZKP proof such that the prover $P$, now
\Tom, proves to the verifier $V$, now \Alice, that he knows a solution of a
given Takuzu grid.  We assume that the grid is printed on a sheet of paper, and
use two kinds of cards: cards with a 0 or a 1 printed on them.  Moreover, we
also use a second grid, a piece of paper, and an envelope.%

\subparagraph*{Setup:}
Let $G$ be the $n\times{}n$ Takuzu grid and $S$ its
solution known by $P$. The prover chooses uniformly at random two permutations:
$\pi_R$ for the rows, and $\pi_C$ for the columns. He then computes $S'=\pi_R (
\pi_C(S))$, and writes the two permutations $\pi_C$ and $\pi_R$ on a paper, and
inserts it into an envelope $E$.  Then $P$ places cards face down on the second
grid according to $S'$.  We denote these cards $\tilde{S'}$.%

\subparagraph*{Verification:}
The verifier $V$ picks $c$ randomly among $\{0,1,2,3\}$.
\begin{compactdesc}
\item[If $c=0$:] $P$ opens the envelope $E$. Then $V$ takes the initial grid $G$
  and computes $G'=\pi_R ( \pi_C(G))$.  $V$ reveals those cards in $\tilde{S'}$
  that are in places of initial values of the grid in $G'$.  He checks that the
  all revealed cards are equal to the initial values given in $G'$.
\item[If $c=1$:] 
    The verifier $V$ picks $d$ randomly among $\{0,1\}$. 
    If $d=0$ (resp. $d=1$), for each row (resp. column),
    the verifier takes all the cards in the row (resp. in the column) without
    looking at them. Each one of these $n$ decks is shuffled by the shuffle functionality and then
    the verifier opens all cards and checks that each deck contains exactly
    the same number of 1's and 0's.
\item[If $c=2$:]
    The verifier checks that a randomly chosen row or column is different from all 
    other rows or columns. %
    For this, the verifier picks randomly one row or one column.
    The verifier opens all the cards of his chosen row (or column). For each of the 
    $n-1$ other rows (or columns) the verifier takes all the cards that are in the 
    same column (or row) as a 0 in the revealed row (or column), without looking at 
    them. Each one of these $n-1$ sets of cards is shuffled by the
    shuffle functionality and given back to the prover, who reveals one
    card per set that is a 1. Thus each one of the other $n-1$ rows (or columns) has 
    a 1 where the revealed row (or column) has a 0, they are thus different from the 
    revealed row. If there are
    several 1's in a deck, the prover randomly chooses which one to reveal.
\item[If $c=3$:] $P$ opens $E$. Then $V$ permutes (face down) the
  cards of $\tilde{S}'$ to obtain $\tilde{S} = \pi_c^{-1} ( 
  \pi_R^{-1}(\tilde{S'}))$. Then, $V$ picks $d$   randomly among $\{0,1\}$ and  $e$ randomly among $\{1,2,3\}$.
  \begin{compactdesc}
    \item[If $d=0$:] For each row, $V$ sets $x=\lfloor \frac{n-e}{3}  \rfloor$ decks of three cards  $\{(e+ 3\cdot i + 1, e+3\cdot i+2, e+3\cdot i +3)\}_{\{0\leq i < x\}}$  where the triplet $(i,j,k)$ denotes a deck containing the  $i^\text{th}$, the $j^\text{th}$ and the $k^\text{th}$ cards of the row.
    \item[If $d=1$:] For each column, $V$ sets $x=\lfloor \frac{n-e}{3}  \rfloor$ decks of three cards  $\{(e+ 3\cdot i + 1, e+3\cdot i+2, e+3\cdot i +3)\}_{\{0\leq i < x\}}$  where the triplet $(i,j,k)$ denotes a deck containing the  $i^\text{th}$, the $j^\text{th}$ and the $k^\text{th}$ cards of the column.
  \end{compactdesc}

  Then, $V$ gives the  decks to $P$ one by one. For each deck, the prover discards one of
  the two identical cards (\emph{e.g.}, a 1 if he has 101, and a 0 in
  case of \emph{e.g.}, 001). Then $P$ reveals the cards to $V$, who
  accepts only if he sees two different cards. Note that it is
  possible to do this verification step for several sequences of three
  cards.
\end{compactdesc}
To have the best security guarantees, the verifier should choose his challenges $c$, $d$, etc. such that each combination of challenges at the end has the same probability.
This protocol is repeated $\seck$ times where $\seck$ is a chosen
security parameter.  Note that the ZKP is again polynomial in the size
of the grid, as shown next.

\subsection{Security Proofs for Takuzu}%
We now prove the security of our construction. 
\begin{lemma}[Takuzu Completeness]
If $P$ knows a solution of a given Takuzu grid, then he is able to convince $V$.
\end{lemma}
\begin{proof}
Suppose that $P$, knowing the solution $S$ of the grid $G$, runs the setup algorithm as described in Section~\ref{ssec:takuzuzkp}. Then we show that $P$ is able to perform the proof for any challenge $c \in \{0,1,2,3\}$.

\begin{compactdesc}
\item[If $c=0$:] Since $S$ is the solution of $G$, $G'=\pi_R (
  \pi_C(G))$ and $S'=\pi_R ( \pi_C(S))$,  the values of the cards in
  $\tilde{S'}$ that are in places of non empty cells of $G'$  are equal to the corresponding  values of the grid in $G'$.
\item[If $c=1$:] $S$ has $(n/2)$ occurrences of $1$ and $(n/2)$
  occurrences of $0$ on each row and column. Column and
  row permutations do not change this property. Therefore, since $S'=\pi_R (
  \pi_C(S))$, $S'$ has $(n/2)$ occurrences of $1$ and $(n/2)$
  occurrences of $0$ on each row and column.
\item[If $c=2$:] If a row (resp. column) of $S$ is unique and all rows (resp.
  columns) have the same number of 0's, then no other row of $S$ can have its
  0's at the exact same places. Column and
  row permutations in $S'$ do not change this property.
\item[If $c=3$:] Since $S'=\pi_R ( \pi_C(S))$, $\tilde{S} = \pi_C^{-1}
  (\pi_R^{-1}(\tilde{S'}))$ is the solution $S$ hidden with cards face
  down, and three consecutive cells of $S$ are never the same since
  $S$ is a valid solution. Then, using three consecutive cards, the
  prover is always able to discard one out of three cards such that
  two different cards remain.
\end{compactdesc}
\vspace{-\baselineskip}
\end{proof}

\begin{lemma}[Takuzu Soundness]\label{lem:takuzu:sound}
If $P$ does not provide a solution of a given $n\times{}n$ Takuzu grid $G$,
then he is not able to convince $V$ except with a negligible probability 
$\left( 1-\frac{1}{2n+9} \right)^\seck$
when the protocol is repeated $\seck$ times.
\end{lemma}
\begin{proof}
We show that if $P$ is able to perform the proof of a solution of $G$
for any challenge ($c$ and the sub-challenges $d$, $e$, etc. depending on $c$), 
then he knows a solution to the Takuzu
grid. During the setup phase, $P$ commits:
\begin{compactitem}
\item An envelope $E$ containing two permutations $\pi_C$ and $\pi_R$.
\item A grid of face down cards  $\tilde{S'}$. 
\end{compactitem}
We set $S= \pi_C^{-1} (\pi_R^{-1}(S'))$.  Since $P$ is able to perform the
proof for any challenge%
, we observe that:
\begin{compactitem}
\item Non empty cells of $G'=\pi_R ( \pi_C(G))$ are equal to corresponding cells
  of $S'$. Then \textbf{non empty cells of $G$ are equal to corresponding cells
  of $S$}. 
\item Rows and columns of $S'$ have the same number of $0$s and
  $1$s, and each row and each column of $S'$ is unique. 
  Then, \textbf{rows and columns of $S$ have the same number of
  $0$s and $1$s, and each row and each column of $S$ is unique}.
\item \textbf{Three consecutive cells of $S$ do not contain the same value}. 
\end{compactitem}
We deduce that $S$ is a solution of $G$. Hence, if $P$ does not
provide a solution of $G$, then he fails the proof for at least one of
the challenges.
To compute the probability, we enumerate the possible challenges for all values of $c$: 
\begin{compactdesc}
\item[If $c=0$:] In this case there is only one possible challenge.
\item[If $c=1$:] There are two possibilities, verifying the rows, or verifying
  the columns.
\item[If $c=2$:] There are $2n$ choices for the verifier, $n$ rows and $n$ columns.
\item[If $c=3$:] There are two possible values for the challenge $d$ and three possible values for the challenge $e$, then, there is $2\times 3=6$ combinations for the pair $(d,e)$.
\end{compactdesc}
Overall, $P$ receives a challenge out of $1+2+2n+6=2n+9$ possibilities. 
We suppose that the verifier selects any one of these challenges uniformly at
random.
If the prover gives a wrong grid, then at least one of the checks will fail, and
this check will have been selected with probability $\frac{1}{2n+9}$.
As the protocol is repeated $\seck$ times, the probability that $P$ convinces
$V$ without the solution is at least $\left( 1-\frac{1}{2n+9} \right)^\seck$.
\end{proof}

\begin{lemma}[Takuzu Zero-Knowledge]
During the verification phase, $V$ learns nothing about $P$'s solution for a given Takuzu grid.
\end{lemma}

\begin{proof}
Similar to the proof for Akari, we show how to construct a simulator for each challenge $c\in\{0,1,2,3\}$:
\begin{compactdesc}
\item[$c=0$:] The simulator completes the grid $G$ with random bits to obtain the grid $S$, and randomly chooses the two permutations $\pi_R$ and $\pi_C$. It then puts $\pi_R$ and $\pi_C$ in an envelope $E$ and  it commits $S'=\pi_R ( \pi_C(S))$ using cards face down (we denote this commitment $\tilde{S'}$). Then it simulates the prover and the verifier as follows: the prover commits $(E,\tilde{S'})$. Then the verifier can open $E$ and use $\pi_R$ and $\pi_C$ to compute $G'=\pi_R ( \pi_C(G))$.
 $V$ can then return those cards in $\tilde{S'}$ that are in places of initial values of the grid in $G'$ and check that all returned cards are equal to the initial values given in $G'$. Since $\pi_R$ and $\pi_C$ are randomly chosen, and since the values of cards in $\tilde{S'}$ that are in places of non empty cells $G'=\pi_R ( \pi_C(G))$  are equal to the corresponding  values of the grid in $G'$ (the other cards of $\tilde{S}'$ remain face down), the simulated proofs and real proofs are indistinguishable. 
 \item[$c=1$:] 
     When the verifier shuffles the decks of cards taken in a single row or column 
     using the shuffle functionality, the simulator returns $(n/2)$ cards with 1's 
     and $(n/2)$ cards with 0's, shuffled. 
     This is indistinguishable from a shuffled deck with the same number of 1's and 0's.
 \item[$c=2$:] 
     When the verifier selects a row or column, the simulator randomly chooses
     to position $(n/2)$ 1's and $(n/2)$ 0's. Given the random permutations
     $\pi_R$ and $\pi_C$, the permuted chosen row is indistinguishable from a
     randomly chosen one.
     Then, when the verifier shuffles the deck corresponding to selected cards
     in the other rows (or columns), the simulator places one card with a 1, and
     $(n/2-1)$ randomly selected other cards. Once again this is
     indistinguishable from the given decks, as the verifier only sees one card
     with a 1.
 \item[$c=3$:] The simulator chooses randomly $S$ such that three
   cells of $S$ never contain the same bit. It randomly chooses the
   two permutations $\pi_R$ and $\pi_C$ and puts $\pi_R$ and $\pi_C$
   in an envelope $E$. It commits  $S'=\pi_R ( \pi_C(S))$  using cards
   face down (we denote this commitment $\tilde{S'}$). Then it
   simulates the interaction between the prover and the verifier as follows: 
   the prover commits $(E,\tilde{S'})$.
   The verifier opens $E$ and uses
   $\pi_R$ and $\pi_C$ to compute $\tilde{S}= \pi_C^{-1}
   (\pi_R^{-1}(\tilde{S'}))$, where $\tilde{S}$ is the
   commitment of $S$ using face down cards. The verifier then randomly
   chooses $d$ and $e$ and collects the cards according to the
   verification algorithm. 
$V$ then gives each deck of three cards to the prover. 
For each deck, the prover thus obtains three cards
such that (exactly) two cards are identical. He discards one of these
two cards and returns the two different cards.   
Since $\pi_R$ and $\pi_C$ are randomly chosen, then simulated proofs
and real proofs are indistinguishable.  
\vspace{-\baselineskip}
\end{compactdesc}
\end{proof}
Therefore, our protocol for Takuzu is complete, sound and zero-knowledge.

\section{Kakuro}\label{sec:kakuro}

\emph{Kakuro} can been seen as a numerical version of
crosswords. According to~\cite{shortz06}, it was proposed for the
first time in the 1950's as a logic puzzle in ``Official Crossword
Puzzles'' by Dell Publishing Company. It is also known by its the original
English name \emph{Cross Sums}.

A Kakuro grid contains square and triangular white cells, and sometimes also
black cells.  The goal is to fill in the white cells on the grid with digits
from 1 to 9, such that the sum of each line and each column corresponds to the
total given in each triangular cell.  Moreover, in each line and in each column
a number can appear only once.  An example of the game with its the solution is
given in Figure~\ref{fig:kakuro}. This game has been proven NP-complete
in~\cite{journals/icga/KendallPS08,KakuroNP}.

\begin{figure}[bt]
\begin{center}
\begin{tikzpicture}
 \draw[fill=white] (0,0) rectangle (2,2);

 \draw[fill=white] (0,2) -- (0,3) -- (1,2) -- (1,3) -- (2,2) -- (0,2);
 \draw[fill=white] (0,0) -- (-1,1) -- (0,1) -- (-1,2) -- (0,2) -- (0,0);

 \draw[] (0,1) -- (2,1);
 \draw[] (1,0) -- (1,2);

 \draw (-0.25,0.66) node {4};
 \draw (-0.25,1.66) node {3};
 \draw (0.36,2.33) node {4};
 \draw (1.36,2.33) node {3};
\end{tikzpicture}
\qquad
\begin{tikzpicture}
 \draw[fill=white] (0,0) rectangle (2,2);

 \draw[fill=white] (0,2) -- (0,3) -- (1,2) -- (1,3) -- (2,2) -- (0,2);
 \draw[fill=white] (0,0) -- (-1,1) -- (0,1) -- (-1,2) -- (0,2) -- (0,0);

 \draw[] (0,1) -- (2,1);
 \draw[] (1,0) -- (1,2);

 \draw (-0.25,0.66) node {4};
 \draw (-0.25,1.66) node {3};
 \draw (0.36,2.33) node {4};
 \draw (1.36,2.33) node {3};

 \draw (0.5,0.5) node {\Large 3};
 \draw (0.5,1.5) node {\Large 1};
 \draw (1.5,0.5) node {\Large 1};
 \draw (1.5,1.5) node {\Large 2};
\end{tikzpicture}
\end{center}
\caption{Simple example of a Kakuro grid and its solution.}\label{fig:kakuro}
\end{figure}
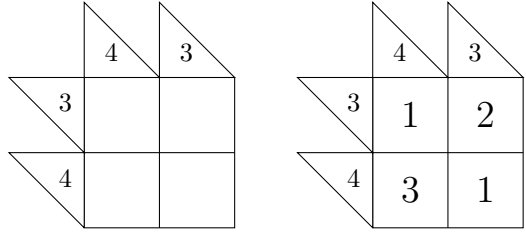

\subsection{ZKP Construction for Kakuro}

In this game, the main challenge is to be able to verify that a sum 
of some numbers is correct without revealing any other information.%

\subparagraph*{Setup:}
In this construction, we use black and red playing cards.
To represent a number $l$, for instance $3$, we put $9$ cards into an envelope: $l=3$ black cards and $9-l=6$ red cards.
Using this trick, we can construct a ZKP as follows:
\begin{compactenum}
\item In each cell, we place 4 identical envelopes containing some
  cards that correspond to the number contained in the cell of the
  solution. The number is encoded using some cards as explained above.
\item Next to each triangular cell, we place envelopes containing cards for all
  missing numbers in this line or column. For instance, for the first line of
  Figure~\ref{fig:kakuro}, we place $7$ envelopes containing black and red cards
  corresponding to the numbers $3$, $4$, $5$, $6$, $7$, $8$, $9$.
\end{compactenum}

\subparagraph*{Verification:}
We proceed as follows:
\begin{compactitem}
\item 
To verify that a number appears only once per line and per
column, the verifier proceeds as follows: For each line and each column,
he randomly picks an envelope per cell plus all the envelopes next to the
triangular cell. The packet of envelopes is shuffled by the shuffle 
functionality, and then the verifier opens
all the envelopes and verifies that all numbers appear only once, and
that all numbers between $1$ and $9$ are present.
\item 
To verify that the sum per line and per column corresponds to the
number in the triangular cell, the verifier randomly picks one envelope per
cell in the line (resp. in the column) and opens (face down) the content of each
envelope. All the cards are shuffled by the shuffle functionality and then the
verifier returns them. He can then 
check that the number of black cards corresponds to the number given
in the triangular cell.
\end{compactitem}
This protocol is repeated $\seck$ times where $\seck$ is a chosen
security parameter.  It is easy to see that the ZKP is polynomial in
the size of the grid as we only need a polynomial number of cards and
verification steps.

 \subsection{Security Proofs for Kakuro}\label{app:kakuro}

We prove the security of our ZKP protocol for Kakuro.
\begin{lemma}[Kakuro Completeness]\label{lem:kakuro:complete}
If $P$ knows a solution of a given Kakuro grid, then he is able to convince $V$.
\end{lemma}
\begin{proof}
After a correct setup each cells contains identical envelopes.
\begin{compactitem}
\item The unicity of numbers per row or column is given by the fact
  that all $n$ numbers are present exactly once between the envelopes
  within the cell and next to the triangle of the row/column.
\item The correctness of the sum is given by the fact that when mixed
  the number of black cards is exactly the value within the triangle.
  \end{compactitem}
\vspace{-\baselineskip}
\end{proof}

\begin{lemma}[Kakuro Soundness]\label{lem:kakuro:sound}
If $P$ does not provide a solution of a given Kakuro grid,
then he is not
able to convince $V$ except with a negligible probability lower than $p=(1/4)^\seck$ 
when the protocol is repeated $\seck$ times.
\end{lemma}
\begin{proof}
The proof follows the same line as
in~\cite[Lemma~1]{Gradwohl:2007:CPZ:1760607.1760623}.  Each rule is
separately verified (unicity in row/column, correct sum in
row/column).

Assume that the prover does not know a valid solution for the
puzzle. Then he is always caught by the protocol as a liar if he
places the cards such that each cell has four cards of identical
value. 
The only way a prover can cheat is by placing on a cell, say
cell $a$, four envelopes that do not all contain the same value.
This means that in this cell at least one envelope contains a value $y$ different
from at least $2$ of the other $3$ numbers.
Given any assignment of envelopes to the unicities and sums for all 
other cells, there is either only one envelope with value $y$ in the cell and 
thus for the (cheating) prover exactly one of the unicities/sums rules
that requires $y$ in this cell, or there are two envelopes with value $y$ in the cell and
exactly two of the unicities/sums rules require $y$ in this cell.
In the first case, the probability that for cell $a$ the verifier
assigns $y$ to the one rule needing it is $1/4$.
When there are two envelopes encoding $y$ in the cell, the probability to assign $y$ to
the first rule needing it is $1/2$ and then the probability to assign
the second $y$ to the second rule needing it is $1/3$, overall this
is~$1/6<1/4$. 
As the protocol is repeated $\seck$ times, the probability that $P$ convinces $V$ without the solution is bounded by $(1/4)^\seck$, which is negligible.
\end{proof}

\begin{lemma}[Kakuro Zero-Knowledge]\label{lem:kakuro:zk}
$V$ learns nothing about $P$'s solution of a given Kakuro grid.
\end{lemma}
\begin{proof}
As in the proof for Akari (Lemma~\ref{lem:akari:zk}) we use the
advantage of the simulator over the prover: when shuffling packets (of selected
envelopes to show unicity; or collected cards for sums) it is allowed to
swap the packets for different ones. The simulator acts as follows:
\begin{compactitem}
\item The simulator places four arbitrary envelopes on each cell.
\item The verifier randomly picks the envelopes for the
  corresponding packets.
\item Then there are two cases:
\begin{compactitem}
\item When verifying the unicity of a number, the envelopes in a packet are shuffled. The
  simulator swaps the packets with a randomly shuffled packet of
  envelopes, in which each value appears once.
\item When verifying the sums, the envelopes are opened and mixed by the verifier. 
 Then all the cards are shuffled. The simulator swaps the set of cards
 with a randomly shuffled packet with the same number of cards, in
 which the number of black cards corresponds to the sum.
 \end{compactitem}
\end{compactitem}
The final packets are indistinguishable from those provided by an
honest prover assuming that the shuffle functionality
guarantees that the packets each contain randomly shuffled sets.
\end{proof}

\section{Ken-Ken}\label{sec:kenken}

\emph{KenKen} is a Japanese game invented by Tetsuya Miyamoto, also
known as \emph{Calcudoku}, \emph{Mathdoku} or \emph{Kendoku}.  It
combines ideas from Sudoku and Kakuro. A KenKen grid is a square grid
of size $n \times n$. To solve a KenKen grid, like in Sudoku, each row
and each column must contain exactly once all numbers between $1$ and
$n$. Moreover, a KenKen grid is divided in groups of cells called
\emph{cages}. The example given in Figure~\ref{fig:kenken} contains 3
cages: one vertical line of 3 cells, one horizontal line of 2 cells
and a square of 4 cells. Each cage contains a \emph{target} number
that has to be produced using the numbers in the cells (in any order)
and the mathematical operation (addition, subtraction, multiplication
or division) given after the target number. For example in
Figure~\ref{fig:kenken}, the vertical three-cell cage with the
addition operator and a target number of 6 is satisfied with the
digits $1$, $2$, and $3$, as $3+1+2=6$. The target number and the operator
are given in the upper left corner of a cage.  The target must be a
positive integer. KenKen is known to be
NP-complete~\cite{Koelker12,DBLP:journals/jip/HaraguchiO15,DBLP:journals/ieicet/HaraguchiO13}.

In most KenKen grids, division and subtraction operators are restricted
to cages of only two cells.  For instance in the grid given in
Figure~\ref{fig:kenken}, the cage with the subtraction operator and
the target of 1 has four possible solutions when analyzed in
isolation: $(1,2)$, $(2,1)$, $(2,3)$ or $(3,2)$. In general grids do not need
to respect this hypothesis, and there are puzzles that use more than
two cells for these operations.  In any case, for each subtraction or
division cage, there is at least one maximal element of which the
other elements in the cage are subtracted or divided.

\begin{figure}[bt]
\begin{center}
\begin{tikzpicture}
 \draw[line width=2pt] (0,0) rectangle (3,3);

 \draw[] (2,0) -- (2,3);
 \draw[] (0,1) -- (3,1);
 \draw[] (0,2) -- (3,2);

 \draw[line width=2pt] (1,0) -- (1,3);
 \draw[line width=2pt] (1,2) -- (3,2);
 
 \draw (0.3,2.75) node {\tiny $+$ 6};
 \draw (1.3,2.75) node {\tiny $-$ 1};
 \draw (1.3,1.75) node {\tiny $*$ 18};
\end{tikzpicture}
\qquad
\begin{tikzpicture}
 \draw[line width=2pt] (0,0) rectangle (3,3);

 \draw[] (2,0) -- (2,3);
 \draw[] (0,1) -- (3,1);
 \draw[] (0,2) -- (3,2);

 \draw[line width=2pt] (1,0) -- (1,3);
 \draw[line width=2pt] (1,2) -- (3,2);
 
 \draw (0.3,2.75) node {\tiny $+$ 6};
 \draw (1.3,2.75) node {\tiny $-$ 1};
 \draw (1.3,1.75) node {\tiny $*$ 18};

 \draw (0.6,2.4) node {\Large 3};
 \draw (1.6,2.4) node {\Large 1};
 \draw (2.6,2.4) node {\Large 2};
 \draw (0.6,1.4) node {\Large 1};
 \draw (1.6,1.4) node {\Large 2};
 \draw (2.6,1.4) node {\Large 3};
 \draw (0.6,0.4) node {\Large 2};
 \draw (1.6,0.4) node {\Large 3};
 \draw (2.6,0.4) node {\Large 1};
\end{tikzpicture}
\end{center}
\caption{Simple example of a KenKen grid and one possible solution.}\label{fig:kenken}
\end{figure}

\subsection{ZKP Construction for KenKen}
We use the same idea as in
Sudoku~\cite{Gradwohl:2007:CPZ:1760607.1760623} to verify that all
numbers appear only once per row and column.  We also use the same
representation of numbers as in Kakuro, and the same technique to
verify the addition cages.  For multiplications, the idea is to check
the sum of the exponents of the prime factors of the target.  Indeed,
in a $n\times{}n$ grid, all prime factors are below $n$ and there are
at most ${\mathcal O}(n/\log(n))$ of those~(see, \emph{e.g.},
\cite{Rosser:1962:formulas}), so we need at most that number of
parallel exponent addition protocols.  There we need to encode each
integer by its prime factor exponents: for this we use small envelopes
marked with the prime factor $p$, called $p$-envelopes, and containing
the black/red cards encodings for the associated exponent.
The
maximum possible exponents have to be found among the factors of the
integers between $1$ and $n$.  
For instance, with $n=9$, all the integers are encoded with the primes
$2$, $3$, $5$, $7$, with respective exponents between $0$ and $3$,
between $0$ and $2$, and between $0$ and $1$ for both $5 $ and $7$. 
Therefore the $2$-envelopes  contain exactly $3$ cards, the
$3$-envelopes $2$ cards and the $5$- and $7$-envelopes, $1$ card. 

To deal with subtractions and
divisions, we need an extra interactive round to identify the largest
integer in a cage, and then we reduce to either an addition or a
multiplication verification.  To remain zero-knowledge even after the
identification of the maximal element, the solution of the cage is
mixed with other solutions, with all the other possible maximal
elements, before the verifier checks them.  Finally, to be able to
deal with both additions/subtractions and multiplications/divisions,
we  need larger envelopes containing both kind of encodings per
cell. The encodings must match.  For instance, a $6$ in a grid of size
$9$ is  encoded by a large envelope containing: $6$ black cards and $3$
red cards (just like for Kakuro); but also a small $2$-envelope marked
with a $2$ and containing $1$ black card and $2$ red cards; similarly,
a small $3$-envelope with $1$ black card and $1$ red card; a small
$5$-envelope with $1$ red card and a small $7$-envelope with $1$ red
card.  This works also for the value $1$, encoded with $p$-envelopes
containing only red cards.%

\subparagraph*{Setup:}
Our ZKP scheme works as follows for a grid of size~$n$:
\begin{compactitem}
\item In each cell, we place three identical envelopes encoding the
  number of the solution, in both encodings.
\item For each subtraction cage with~$c\geq{}2$ cells of target~$t$,
  let $max$ be the maximal value in the $c$ cells and $c_i\geq{}1$,
  $i=1,\ldots,c-1$, the other values in any order. Then
  $t=max-\sum_{i=1}^{c-1} c_i$ and $n\geq{}max=t+\sum_{i=1}^{c-1} c_i\geq{}t+(c-1)$.
  We thus place at most~$(n-t-c+1)$ large envelopes next to
  the grid. Each of these envelopes contains: one marked small envelope 
  (itself containing~$n$ cards) and~$n(c-1)$ other
  black or red cards. The number of black cards in the small
  envelope minus the number of other black cards must match the
  target. These~$(n-t-c+1)$ large envelopes contain all combinations
  (corresponding to distinct maximal elements) except the one from the solution.
\item 
  Each division cage with~$c$ cells of target~$t$ must contain a
  maximal element divisible by $t$ and by the $c-1$ other elements.
  This maximal element must be less or equal than $n$ and larger or
  equal than $t$.  We denote by $u$ the number of possible maximal
  elements, $u=|\{m, n\geq{}m\geq{}t~\text{and}~m/t\in\N\}|\leq(n-t+1)$.
  Then we place~$u-1$ large envelopes next to the grid.  Each of these
  large envelopes contains: a full set of~$p$-envelopes and another
  envelope, marked, itself containing a full set of~$p$-envelopes, for
  all primes $p\leq{}n$.  For each prime~$p$, the number of black
  cards in the~$p$-envelopes of the marked envelope minus the number
  of black cards in the other~$p$-envelope equals the exponent of~$p$
  within the factorization of the target~$t$. For instance if the
  target is~$4$ and the~$2$-envelope within the marked envelope
  contains~$3$ black cards, then the other~$2$-envelope in the
  large envelope must contain exactly~$1$ black card.  A
  complete example for a division cage is given below.
\end{compactitem}
\subparagraph*{Verification:}
Once all placements are done, the verifier
randomly picks envelopes placed on each cell and perform the
following verifications:
\begin{compactitem}
\item To verify that each number appears only once per row and
  column, the verifier randomly picks for each line and for each column one
  envelope per cell. This packet is shuffled by the shuffle functionality, and then
  the verifier  opens all the envelopes and check that all
  numbers appear exactly once. Moreover, the verifier  checks that
  both encodings coincide, that is that the exponents within the
  $p$-envelopes coincide with the factorization of the number of black cards.
\item For each addition cage, the verifier randomly picks one envelope per
  cell. Then he opens all the envelopes and discards (or gives
  back to the prover) all its~$p$-envelopes without
  opening them. Finally, as in Kakuro, the verifier mixes the black
  and red card from all envelopes face down, 
  asks the shuffle functionality to shuffle them and then verifies that the
  number of black cards corresponds to the target number.
\item For each subtraction cage, there is an extra interactive round:
  \begin{compactenum}
  \item The verifier randomly picks one envelope per cell in the cage, and asks
    the shuffle functionality to shuffle them;
  \item The verifier gives to the prover these envelopes, one at a
    time, after discarding (or giving back to the prover);
  \item The prover looks inside, and has two possibilities:
    \begin{compactitem}
    \item If the envelope does not contain the maximum of the cage, the prover
      gives back the envelope unmodified to the verifier.
    \item Otherwise the prover marks the envelope (in view of the
      verifier, for instance with a pencil) and gives the marked
      envelope back to the verifier.
    \end{compactitem}
    If there are multiple copies of a maximum element in the cage, the
    prover randomly chooses which one to mark.
  \item The verifier empties all the envelopes, but the one
    marked by the prover, into a larger envelope (all these cards
    are shuffled using the functionality) and discard all the~$p$-envelopes.
  \item The verifier adds the marked envelope, still sealed, to the same
    larger envelope.
  \item The verifier asks the shuffle functionality to shuffle this
    larger envelope with the~$(n-t)$ other large envelopes associated to the
    subtraction cage. Then the verifier opens all large envelopes,
    checks that each large envelope satisfies the target (black cards
    in the marked envelope minus the free black cards in the large
    envelope equals the target) and checks that the~$n-t+1$ possible
    combinations are present exactly once. 
  \end{compactenum}
\item For each multiplication cage of target~$t$, the verifier
randomly picks one envelope per cell and opens them all.
He discards (or gives back to the prover) the free black and red
cards without returning them. Then, one prime~$p$ at a time, he empties
all the associated small~$p$-envelopes, mixes all the black and red
cards, asks the shuffle functionality to shuffle them and then
verifies that 
the number of black cards is the exponent of the prime factor~$p$ in
the factorization of~$t$.
\item 
  For each division cage, we mix the subtraction and
  multiplication protocols: as in the subtraction, the prover and the
  verifier enter an interactive extra round to mark an envelope
  containing a maximal element; then for each of the $u$ possible
  maximal elements, there is a set of possible
  multiplicative solutions.
  There is a complete example below. %
\end{compactitem}
This protocol is repeated $\seck$ times where $\seck$ is a chosen
security parameter.  The protocol can be verified in \emph{polynomial
  time}.  This stems from the fact that the prime factors are all
lower than $n$.  Therefore, even factoring the target numbers is just
looking at greatest common divisors between $t$ and values from $2$ to
$n$.

\subsection{Example of a division cage setup for KenKen}\label{app:ex:kenken}
To illustrate our construction, we use the division cage with $c=4$ cells given in Figure~\ref{fig:ex:kenken}.
\begin{figure}[tb]
\begin{center}
\begin{tikzpicture}
 \draw[line width=2pt] (0,0) rectangle (2,2);

 \draw[] (1,0) -- (1,2);
 \draw[] (0,1) -- (2,1);

 \draw (0.3,1.75) node {\tiny $\div$ 2};
\end{tikzpicture}
\qquad
\begin{tikzpicture}
 \draw[line width=2pt] (0,0) rectangle (2,2);

 \draw[] (1,0) -- (1,2);
 \draw[] (0,1) -- (2,1);

 \draw (0.3,1.75) node {\tiny $\div$ 2};

 \draw (0.6,1.4) node {\Large 1};
 \draw (1.6,1.4) node {\Large 6};
 \draw (0.6,0.4) node {\Large 3};
 \draw (1.6,0.4) node {\Large 1};
\end{tikzpicture}

\end{center}
\caption{A $2\times{}2$ division cage and one solution within a $9\times{}9$ KenKen
  grid.}\label{fig:ex:kenken}
\end{figure}
 Suppose the cage in the figure is part of an $9\times{}9$ grid and
 that the solution is the one given, that is $2 = 6 /3/1/1$.  As
 $9=3^2$ and $8=2^3$, the maximal exponents  for
 $2$, $3$, $5$ and $7$, will be bounded by $3$, $2$, $1$
 and $1$, respectively, and denoted by $e_2$, $e_3$, $e_5$ and $e_7$, respectively.
   There are $u=4$ possible maximal elements ($2$, $4$,
 $6$, and $8$) because $n=9$ and the target number is $2$.  Moreover, as
 this cage contains $c=4$ cells, the maximal elements are
 divided by $3$ numbers. For instance, not counting orders, the target
 can be obtained with the following solutions (one per possible
 maximum): $2=8/2/2/1$; $2=6/3/1/1$; $2=4/2/1/1$; $2=2/1/1/1$.

\subparagraph*{Setup:}
For each cell of the initial cage the prover places, according to his
solution, the following envelopes, where the number of cards contained
in a $p$-envelope is $e_p$:
\begin{enumerate}
\item For each one of the two $1$'s he places three identical envelopes containing each:
  \begin{itemize}
  \item To verify addition and subtraction: $1$ black
    card and $8$ red cards,
  \item To verify multiplication and division: a $2$-envelope with
    $e_2=3$ red cards, a $3$-envelope with $e_3=2$ red cards, a
    $5$-envelope with $e_5=1$ red card, a $7$-envelope with $e_7=1$ red
    cards;
  \end{itemize}
\item  For the $3$ he places three identical envelopes containing each:
  \begin{itemize}
  \item To verify addition and subtraction verification: $3$ black
    cards and $6$ red cards,
  \item To verify multiplication and division: 
  a $2$-envelope with $e_2=3$ red cards, 
  a $3$-envelope with $1$ black card and $e_3-1=1$ red card, 
  a $5$-envelope with $e_5=1$ red card, 
  a $7$-envelope with $e_7=1$ red card;
  \end{itemize}
\item For the $6$: three identical envelopes containing each:
  \begin{itemize}
  \item To verify addition and subtraction: $6$ black
    cards and $3$ red cards,
  \item To verify multiplication and division: 
  a $2$-envelope with $1$ black card and $e_2-1=2$ red cards, 
  a $3$-envelope with $1$ black card and $e_3-1=1$ red card, 
  a $5$-envelope with $e_5=1$ red card, 
  a $7$-envelope with $e_7=1$ red card; 
  \end{itemize}
\end{enumerate}
Furthermore, the prover prepares three large envelopes next to the grid:
\begin{enumerate}
\item One for the maximal element $8$, containing the encoding of $2^1
  \cdot 2^1 \cdot 1$:
  \begin{itemize}
    \item A $2$-envelope that contains $3 \cdot e_2=9$ cards including
      $2$ black and $ 3 \cdot e_2 -2 =7$ red cards.
\item A $3$-envelope that contains $3 \cdot e_3=6$ cards including $0$
  black card and $ 3 \cdot e_3 $ red cards.
\item A $5$-envelope that contains $3 \cdot e_5=3$ cards including $0$
  black card and $ 3 \cdot e_5 $ red cards.
\item A $7$-envelope that contains $3 \cdot e_7=3$ cards including $0$
  black card and $ 3 \cdot e_7 $ red cards.
    \item A marked envelope containing the encoding of $8=2^3\cdot{}3^{0}\cdot{}5^0\cdot{}7^0$:
      \begin{itemize}

 \item A $2$-envelope that contains $e_2=3$ cards including $3$ black
   cards and $0$ red card.
 \item A $3$-envelope that contains $e_3=2$ cards including $0$ black
   card and $2$ red cards.
    \item A $5$-envelope that contains $e_5=1$ cards including $0$
      black card and $1$ red card.
    \item A $7$-envelope that contains $e_7=1$ cards including $0$
      black card and $1$ red card.
      \end{itemize}
  \end{itemize}
\item One for the maximal element $4$, containing the encoding of $2^1
  \cdot 1 \cdot 1$:
  \begin{itemize}
    \item A $2$-envelope that contains $3 \cdot e_2=9$ cards including
      $1$ black card and $ 3 \cdot e_2 -1=8 $ red cards.
\item A $3$-envelope that contains $3 \cdot e_3=6$ cards including $0$
  black card and $ 3 \cdot e_3 $ red cards.
\item A $5$-envelope that contains $3 \cdot e_5=3$ cards including $0$
  black card and $ 3 \cdot e_5 $ red cards.
\item A $7$-envelope that contains $3 \cdot e_7=3$ cards including $0$
  black card and $ 3 \cdot e_7 $ red cards.
    \item A marked envelope containing the encoding of $4=2^2\cdot{}3^{0}\cdot{}5^0\cdot{}7^0$:
      \begin{itemize}

 \item A $2$-envelope that contains $e_2=3$ cards including $2$ black
   cards and $1$ red card.
 \item A $3$-envelope that contains $e_3=2$ cards including $0$ black
   card and $2$ red cards.
    \item A $5$-envelope that contains $e_5=1$ cards including $0$
      black card and $1$ red card.
    \item A $7$-envelope that contains $e_7=1$ cards including $0$
      black card and $1$ red card.
      \end{itemize}
    \end{itemize}\item One for the maximal element $2$, containing the encoding of $1
  \cdot 1 \cdot 1$:
  \begin{itemize}
    \item A $2$-envelope that contains $3 \cdot e_2=9$ cards including
      $0$ black card and $ 3 \cdot e_2  $ red cards.
\item A $3$-envelope that contains $3 \cdot e_3=6$ cards including $0$
  black card and $ 3 \cdot e_3 $ red cards.
\item A $5$-envelope that contains $3 \cdot e_5=3$ cards including $0$
  black card and $ 3 \cdot e_5 $ red cards.
\item A $7$-envelope that contains $3 \cdot e_7=3$ cards including $0$
  black card and $ 3 \cdot e_7 $ red cards.
    \item A marked envelope containing the encoding of $2=2^1\cdot{}3^{0}\cdot{}5^0\cdot{}7^0$:
      \begin{itemize}
 \item A $2$-envelope that contains $e_2=3$ cards including $1$ black
   card and $2$ red cards.
 \item A $3$-envelope that contains $e_3=2$ cards including $0$ black
   card and $2$ red cards.
    \item A $5$-envelope that contains $e_5=1$ cards including $0$
      black card and $1$ red card.
    \item A $7$-envelope that contains $e_7=1$ cards including $0$
      black card and $1$ red card.
      \end{itemize}
    \end{itemize}
  
\end{enumerate}

\subparagraph*{Verification:}
The cards that are intended for addition or subtraction are discarded.
The verifier and the prover start the round to mark the envelope of
the $6$.  Then the verifier merges the remaining $p$-envelopes.  Thus,
the large envelope  contains at the end:
\begin{itemize}
\item A $2$-envelope that contains $3 \cdot e_2=9$ cards including $0$
  black card and $9$ red cards, encoding the sum of
  exponents of $2$ for $3^1 \cdot 1 \cdot 1$;
\item A $3$-envelope that contains $3 \cdot e_3=6$ cards including $1$
  black card and $5$ red cards, encoding the sum of
  exponents of $3$ for $3^1 \cdot 1 \cdot 1$;
  \item A $5$-envelope that contains $3 \cdot e_5=3$ cards including $0$
  black card and $3$ red cards, encoding the sum of
  exponents of $5$ for $3^1 \cdot 1 \cdot 1$;
  \item A $7$-envelope that contains $3 \cdot e_7=3$ cards including $0$
  black card and $3$ red cards, encoding the sum of
  exponents of $7$ for $3^1 \cdot 1 \cdot 1$;
\item A marked envelope containing the encoding of $6=2^1\cdot{}3^{1}\cdot{}5^0\cdot{}7^0$:
  \begin{itemize}
 \item A $2$-envelope that contains $e_2=3$ cards including $1$ black
   card and $2$ red cards.
 \item A $3$-envelope that contains $e_3=2$ cards including $1$ black
   card and $1$ red card.
    \item A $5$-envelope that contains $e_5=1$ cards including $0$
      black card and $1$ red card.
    \item A $7$-envelope that contains $e_7=1$ cards including $0$
      black card and $1$ red card.
  \end{itemize}
\end{itemize}
This large envelope is then shuffled by the shuffle functionality with
the $3$ other large envelopes prepared in the setup phase.  Then the
verifier  checks that those $4$ envelopes encode one possible cage
solution for each maximum number, namely, $8$, $6$, $4$ and $2$.  That
is, in each of those four large envelopes, for each prime $p$, the
differences in terms of black cards between both $p$-envelopes
always gives the prime factor decomposition of the target, $2$.

\subsection{Security Proofs for KenKen}\label{app:kenken}

Now we prove the security of our algorithm.
 
\begin{lemma}[KenKen Completeness]
If $P$ knows a solution of a given KenKen grid, then he is able to convince $V$.
\end{lemma}
  \begin{proof}
Unicity in rows/columns as well as correctness of addition cages
follows from the completeness of the Kakuro
protocol in Lemma~\ref{lem:kakuro:complete}. 
Correctness for the subtraction comes from the fact that if $k$ is a
maximal element in a cage, then the sum of all the remaining elements
in the cage is equal to $k-t$ (note that this implies
that $k$ must be larger than $t$). 
Correctness for the multiplication of target $t$ is guaranteed for
each prime $p$:
when mixed, the number of black cards in all the $p$-envelopes
is exactly the exponent of $p$ in the factorization of the target $t$.
Similarly, for the division, for each possible maximal element $k$, 
the verifier checks by multiplications that the set of factors always
yield $k/t$.
\end{proof}
In order for the protocol to be acceptable its verification phase
should at least remain polynomial with the size of the grid. We show
that this is indeed the case in Lemma~\ref{lem:kenken:complexity}.
\begin{lemma}[KenKen Complexity]\label{lem:kenken:complexity}
The number of operations to verify a KenKen grid is polynomial in the
size $n$ of the grid.
\end{lemma}
\begin{proof}
It is easy to see that the addition and subtraction protocols are linearly
checked. 
For the unicity checks, as well as for the multiplication and division
cages, we have to consider the number of
distinct possible prime factors. By classical number theory (see for
instance, \cite[(2.18)]{Rosser:1962:formulas}) the number of prime
factors below $n$ is ${\mathcal O}(n/\log(n))$. 
So checking a multiplication cage can be done with that number of prime
exponent checks, just like checking the correspondence of the
encodings in each cell, and checking a division is performed with at most $n$
multiplication checks (one for each possible maximal element in a
cage).
Then each target $t$ is at most $(n!)^n$, so each exponent is at most
$\log_2(t)\leq{}n\log_2(n!)\leq{}n^2\log_2(n)$. Finally the number of
cages in a grid is at most $n^2$, so a very rough bound on the number
of operations for the verifier is $n^{5+o(1)}$.
\end{proof}

\begin{lemma}[KenKen Soundness]
If $P$ does not provide a solution of a given KenKen grid,
then he is not
able to convince $V$ except with a negligible probability $p=(1/3)^\seck$
when the protocol is repeated $\seck$ times.
\end{lemma}
\begin{proof}
As for Kakuro, separately checking unicity rules, addition or
multiplication is perfectly sound.
For subtraction, the prover could mark an element of the
cage which is not maximal. 
But then the subtraction would yield a negative result,
necessarily different from the target. Therefore checking the
subtraction alone is also perfectly sound. A similar argument works
for the division.
Therefore, in a similar way as for Kakuro 
(Lemma~\ref{lem:kakuro:sound}) 
the prover is always caught by the protocol as a liar if he places the
large envelopes such that each cell has three identical envelopes.
Thus again,
the only way a cheating prover can cheat is by placing on a cell, say
cell $a$, three envelopes that do not all contains the same value. 
Then at least
one value is distinct form the other two, and the probability to
randomly pick it for the rule needing it is lower than $1/3$.
As the protocol is repeated $\seck$ times, the probability that $P$ convinces $V$ without the solution is bounded by $(1/3)^\seck$, which is negligible.
\end{proof}

\begin{lemma}[KenKen Zero-Knowledge]\label{lem:kenken:zk}
$V$ learns nothing about~$P$'s solution of a given KenKen grid.
\end{lemma}
\begin{proof}
The same kind of simulator as for Lemma~\ref{lem:kakuro:zk} can be
used: for unicity and addition rules the simulator is directly that of
Kakuro. 
Similarly, for multiplication, subtraction and division cages, the
verifiers look at the colors of the cards only during the final step
of the verification phase, and only after a last shuffle by the prover.
Therefore, in this penultimate step, the simulator can use its ability
to replace shuffles by his choice of envelopes that satisfy the
expected rule. Once again, this indistinguishable from those provided
by an honest prover.
\end{proof}

\section{Conclusion}\label{sec:conclusion}

In this paper, we devised a solution that allows \Charlie to solve his
friend's dilemma for both games. \Alice can now prove to \Tom that she
knows a solution to his Akari problem without revealing his solution,
and \Tom can prove to \Alice that he knows a solution to his Takuzu
problem without revealing it either.
The same is true for \Ken and \Kakarotto playing Kakuro and KenKen.

We showed that our solutions are
secure, that is complete, sound and zero-knowledge. Moreover, we
do not use any cryptographic primitives, but only cards, paper and envelopes.

As future work, we would like to investigate other similar games.
For example, we would like to analyse Futoshiki, which can be seen has a
variation of Sudoku with additional constraints on the order of the
numbers, or Hitori, which has the constraint that all unmarked cells need to be connected, unlike any constraint in the games analysed in this paper.

\bibliography{biblio-short}

\end{document}